\documentclass[]{interact}

\usepackage{hyperref}
\usepackage{caption}
\usepackage{subcaption}

\usepackage{algorithm}
\usepackage{algcompatible}
\usepackage{algorithmicx}
\usepackage{algpseudocode}
\usepackage{amsmath}
\usepackage{comment}
\usepackage{enumitem}

\usepackage[round]{natbib}
\usepackage{epstopdf}

\usepackage{array}
\usepackage{tabularx}
\usepackage{multirow}

\usepackage{xcolor}

\makeatletter

\newcommand{\Rmnum}[1]{\expandafter\@slowromancap\romannumeral #1@}
\makeatother

\newenvironment{practitionersummary}{%
    \renewcommand{\abstractname}{PRACTITIONER SUMMARY}
    \begin{abstract}
}{
    \end{abstract}
}

\newenvironment{wordcount}{%
    \renewcommand{\abstractname}{WORD COUNT}
    \begin{abstract}
}{
    \end{abstract}
}

\theoremstyle{plain}
\newtheorem{theorem}{Theorem}[section]

\newtheorem{proposition}[theorem]{Proposition}

\theoremstyle{definition}

\theoremstyle{remark}

\begin{document}


\title{Optimizing View Change for Byzantine Fault Tolerance in Parallel Consensus}

\author{
\name{Yifei Xie\textsuperscript{a}\thanks{CONTACT Yifei Xie. Email: yifei.xie@ed.ac.uk. School of Informatics, University of Edinburgh, 10 Crichton St, Edinburgh EH8 9AB, UK.}, Btissam Er-Rahmadi\textsuperscript{b}, Xiao Chen\textsuperscript{c}, Tiejun Ma\textsuperscript{a}, Jane Hillston\textsuperscript{a}}
\affil{\textsuperscript{a}School of Informatics, University of Edinburgh, Edinburgh, UK; 
\textsuperscript{b}Edinburgh Research Center, Huawei Technologies R\&D, Edinburgh, UK; 
\textsuperscript{c}School of Computing and Mathematical Sciences, University of Leicester, Leicester, UK}
}

\maketitle

\begin{abstract}
The parallel Byzantine Fault Tolerant (BFT) protocol is viewed as a promising solution to address the consensus scalability issue of the permissioned blockchain. One of the main challenges in parallel BFT is the view change process that happens when the leader node fails, which can lead to performance bottlenecks. Existing parallel BFT protocols typically rely on passive view change mechanisms with blind leader rotation. Such approaches frequently select unavailable or slow nodes as leaders, resulting in degraded performance. To address these challenges, we propose a View Change Optimization (VCO) model based on mixed integer programming that optimizes leader selection and follower reassignment across parallel committees by considering communication delays and failure scenarios. We applied a decomposition method with efficient subproblems and improved benders‘ cuts to solve the VCO model. Leveraging the results of improved decomposition solution method, we propose an efficient iterative backup leader selection algorithm as views proceed. By performing experiments in Microsoft Azure cloud environments, we demonstrate that the VCO-driven parallel BFT outperforms existing configuration methods under both normal operation and faulty condition. The results show that the VCO model is effective as network size increases, making it a suitable solution for high-performance parallel BFT systems.
\end{abstract}

\begin{practitionersummary}
    Scalability in permissioned blockchains is increasingly achieved through parallel BFT protocols, yet maintaining performance during leader failures remains a significant challenge. Existing systems predominantly employ passive view change mechanisms that rotate leadership based on fixed, blind schedules. This approach is inefficient, as it frequently selects unavailable or slow servers, leading to significantly higher latency and preventing immediate transaction progress. To address this bottleneck, this paper presents a View Change Optimization (VCO) model that minimizes the latency of the view change process while ensuring the safety. Formulated as a mixed integer programming problem, the VCO model moves beyond blind rotation to actively identify optimal leader configurations across parallel committees by node-to-node communication delays and failure probabilities. We implement the solution of the model using an improved Benders decomposition technique and integrate it with an efficient iterative leader replacement algorithm as views proceed. Extensive experiments conducted on Microsoft Azure cloud test bed demonstrate that the VCO-driven parallel BFT significantly outperforms standard parallel consensus mechanisms, particularly as network size increases and faulty scenarios. For practitioners designing high-performance BFT consensus protocol in distributed systems, this research offers a promising solution for minimizing latency and performance improvement for the view change process. Our proposed model prioritizes performance metrics over random rotation for leader selection while strictly maintaining BFT safety. This approach offers a practical improvement over traditional passive view change schemes for parallel BFT consensus.
\end{practitionersummary}


\section{Introduction}
Byzantine Fault Tolerant (BFT) consensus algorithms, which coordinate communication among distributed servers in the presence of Byzantine failures, have attracted growing attention with the rapid development of blockchain applications \citep{lamport2019byzantine}. These algorithms ensure a consensus even when up to $(\leq \frac{N-1}{3})$ out of $N$ nodes behave maliciously, making them suitable for constructing distributed systems that operate under adversarial or unreliable conditions. In recent years, driven by the rising interest in blockchain technology, BFT algorithms have been adopted across diverse blockchain platforms, supporting applications such as cryptocurrencies \citep{nakamoto2008bitcoin}, supply chain management \citep{dutta2020blockchain}, trade platforms \citep{lin2022blockchain}, and the Internet of Things (IoT) \citep{dorri2017towards}.
As blockchain systems evolve toward higher throughput and lower latency requirements, BFT consensus has gained increasing attention as it performs better than the widely used Proof-of-Work (PoW) \citep{tschorsch2016bitcoin} and Proof-of-Stake (PoS) \citep{gavzi2019proof} mechanisms. However, traditional BFT systems face scalability bottlenecks due to their high communication complexity, i.e., $\mathcal{O}(N^2)$ with respect $N$ participating nodes. To overcome these limitations, parallel consensus frameworks such as ELASTICO \citep{luu2016secure} and OmniLedger \citep{kokoris2018omniledger} have been proposed to improve the throughput by partitioning the network into multiple smaller committees, each capable of executing BFT consensus in parallel.

Leadership change mechanisms in the BFT system have become critical in BFT systems, driven by the increasing frequency of view changes required in modern deployments. This trend stems from three primary factors: 1) failure handling: as applications scale, the incidence of network problems, bursty workloads, and software bugs rises, necessitating more frequent leader replacement \citep{guo2013failure}; 2) performance requirement: in BFT, malicious leaders can deliberately increase latency without triggering timeouts \citep{amir2008byzantine,aiyer2005bar}, which causes severe system performance degradation. To address this, recent protocols monitor the leader’s throughput and initiate a leadership change if performance drops below a predefined threshold \citep{clement2009making,aublin2013rbft}; 3) prevention of monopolies: frequent leadership changes is essential to limit the influence of malicious nodes that might unjustly prioritize colluding clients for financial advantage \citep{kokoris2018omniledger,zhang2020byzantine}. To mitigate this risk, practical blockchains implement view change protocol \citep{aiyer2005bar}, typically changing leadership every few blocks \citep{moniz2020istanbul,giridharan2023beegees} or at fixed time intervals (e.g., 10 mins) \citep{zhang2021prosecutor}.

As the number of parallel committees grows with the blockchain network scaling, leadership changes occur more frequently since each committee maintains its own leader. Consequently, the efficiency of the view change protocol which governs leader transitions becomes a critical factor in determining overall system performance. Numerous state-of-art parallel BFT algorithms \citep{luu2016secure,kokoris2018omniledger,chen2023parbft,chen2024parallel}, rely on a similar passive view change protocol introduced by PBFT \citep{castro1999practical}. This mechanism is triggered when the current leader node becomes unresponsive or fails to act within the expected time frame --- a view change is initiated to facilitate the election of a new leader. This widely adopted view change protocol follows a strategy to rotate leadership among servers without considering performance issues.

However, this passive view change mechanism is limited by its dependence on a predefined rotation schedule. Since servers must strictly follow a fixed rotation schedule, preventing the system from avoiding unavailable servers during the leader selection process, thereby leading to in weak robustness. For example, if the original leader node fails, the passive view change will assign the scheduled  server. However, if the scheduled server has already crashed, the system can only realize the new leader's unavailability by waiting for timeouts from other servers before the next leader rotation. In addition, the passive view change commonly used can result in inefficiency as it cannot ensure \emph{optimistic responsiveness} \citep{pass2018thunderella}, which requires a non-faulty leader node make immediate progress after being assigned \citep{attiya1994bounds}. Under passive rotation, a slow server may be selected as leader even if it lacks the latest state.
If the server has fallen behind so that it has not yet received the most recent value, it cannot immediately drive consensus. Therefore, the slow server must first synchronize with other servers to update its log, after which it can initiate consensus rounds.
Thus, the passive view change protocol comes at the cost of reduced throughput and increased latency due to a synchronization phase before starting to operate consensus \citep{yin2019hotstuff}.

Some prior research has primarily focused on optimizing parallel BFT protocols under the normal operation case to enhance system performance \citep{chen2023parbft,xie2022stochastic,xie2025stochastic} by mathematical programming models. In \citep{chen2023parbft}, the authors formulated a bi-level mixed-integer linear programming (BL-MILP) model to determine the optimal committee configuration that maximizes throughput. In \cite{xie2025stochastic}, the authors further considered the uncertainty of network conditions and proposed a stochastic programming (SP) model to capture random communication delays and failure rate to optimally configure the committee membership. However, the optimization of the view change process in parallel BFT systems has largely been overlooked, despite its critical role in maintaining system liveness under leader failures. While protocols like Raft \citep{ongaro2014search} enable followers to elect new leaders, and PrestigeBFT \citep{zhang2024prestigebft} introduces a reputation-based scheme to actively trigger view changes when the leader’s performance degrades, these mechanisms are designed for single-committee BFT systems. single-committee BFT systems. To the best of our knowledge, no prior work has investigated the optimization of view-change mechanism in parallel BFT framework.

In order to improve the efficiency of view changes in parallel BFT systems, we propose a View Change Optimization (VCO) model that formulates the leader re-election process as a mixed integer programming (MIP) optimization problem. The VCO model considers communication performance efficiency and fault tolerance to determine the optimal leader configuration across multiple committees. By accounting for node-to-node delay and failure rate, the proposed model minimizes view-change latency and reduces the risk of selecting unreliable nodes. We applied an improved Benders decomposition technique with efficient subproblems and improved Benders' cuts to solve the optimization problem. Then we propose an efficient iterative algorithm for leader replacement as views proceed. Furthermore, we integrate the VCO model with a leader-based parallel BFT algorithm and conduct experiments to evaluate the performance in a real-world cloud network. The experimental results demonstrate that the VCO outperforms non-optimized parallel BFT, with performance improvement  of 23.8\% in normal operation and 45.3\% under faulty condition.

\section{Literature Review} \label{sec:lr}

\subsection{Leader-based BFT in parallel consensus} \label{sec:lr-parbft}
To address scalability and performance challenges, researchers have explored solutions that involve deploying multiple consensus committees. This approach splits the processing of transactions among smaller groups of nodes which can operate consensus in parallel.
Notable examples of such solutions include ELASTICO \citep{luu2016secure}, which employs a hierarchical structure supported by a top-layer committee and several sub-committees, and Omniledger \citep{kokoris2018omniledger}, which utilizes bias-resistant public randomness to generate committees. Additionally, Rapidchain \citep{zamani2018rapidchain} is the first sharding-based public blockchain protocol primarily focused on developing an optimal intra-committee consensus algorithm. ParBFT \citep{chen2023parbft} integrates a sharding scheme with a multi-signature technique to partition the entire consensus network into several committees, allowing transactions to be executed in parallel. Other multi-committee consensus solutions, such as \citep{al2018chainspace,amiri2021sharper}, have been proposed to enhance performance through consensus parallelism.

Most parallel consensus protocols rely on classical BFT for the intra-committee consensus protocol. BFT protocols can broadly be classified into leader-based and leaderless categories. Leaderless BFT algorithms such as \cite{miller2016honey,crain2018dbft,duan2018beat} do not use a designated leader for protocol operation, thereby reducing the risk of single points of failure and server bottlenecks. Although these algorithms are more robust and avoid the leadership changes, they often incur high message complexity and time costs. 
In contrast, leader-based BFT algorithms are favored in practical applications due to their high performance. For example, Hyperledger Fabric \citep{androulaki2018hyperledger} implements PBFT as its consensus protocol, while R3 Corda \citep{hearn2016corda} uses BFT-SMaRt \citep{bessani2014state}. 

To illustrate, we briefly review PBFT \citep{castro1999practical}, a classical leader-based BFT consensus protocol which is commonly used as an intra-committee consensus in sharded blockchains \citep{wang2019sok}. PBFT can tolerate up to one-third Byzantine faults. Its consensus procedures involve a designated leader node determining the order of clients' requests and forwarding them to the follower nodes. Together, all nodes execute a three-phase agreement protocol (pre-prepare, prepare, commit) to reach a consensus on the order of requests. Each node processes every request and sends a response to the respective client. 

\subsection{Committee leadership and view change} \label{sec:lr-vc}

In intra-committee consensus, BFT operation is led by a  committee leader who coordinates the ordering process to reach consensus. When the consensus procedure identifies that the current leader node is faulty, the other nodes trigger a view-change protocol to select a new leader. This mechanism is crucial for maintaining the integrity and reliability of the consensus process, ensuring that a single leader cannot persistently disrupt the system.
In the passive view change introduced by PBFT, servers follow a predefined schedule to rotate leaders. Such practical applications as Aardvark \citep{clement2009making} implement regular view changes whenever a leader's performance degrades beyond a certain threshold, ensuring that no single leader can slow down the system indefinitely. Similarly, HotStuff \citep{yin2019hotstuff} proposes rotating views for each request, aiming to distribute leadership responsibilities more evenly and reduce the risk of prolonged disruptions by any single leader. However, the frequent passive view changes in these protocols can lead to the recurrent selection of faulty leaders, thus having a negative impact on performance. This issue underscores the need for a more appropriate approach to leader selection and rotation.

Raft's \citep{ongaro2014search} leader election mechanism allows servers to actively campaign for leadership when they detect a leader’s failure and vote for an up-to-date node to become a new leader. Consequently, it can prevent unavailable and slow servers from becoming leaders. PrestigeBFT \citep{zhang2024prestigebft} integrates an active view-change protocol with a reputation mechanism in a leader-based BFT consensus algorithm. This innovative approach ensures that only reliable servers can be selected as leaders based on their historical behavior and performance. 
The reputation-driven view-change process in PrestigeBFT improves overall system performance and fault tolerance by dynamically selecting leaders who have demonstrated reliability. This method addresses limitations of traditional passive protocols, which often overlook the varying reliability of different servers. Both Raft and PrestigeBFT focus on optimizing  the view change process within single-committee BFT systems.

Building on the insights and addressing the limitations identified in prior studies on view change mechanism, our work aims to optimize the view change protocol for parallel BFT systems. This involves developing mechanisms that can efficiently manage view change process for multiple committees operating in parallel. 

\subsection{Reliable network optimization under failure} \label{sec:lr-reno}
In network optimization, reliable network design under failures has received considerable attention, particularly in the context of facility location problems, where facilities are placed to minimize transportation costs while accounting for potential disruptions. Several studies have addressed this issue by developing models that incorporate failure scenarios into the strategic planning of hub-based networks, which is closely relevant to optimizing view change in BFT for parallel consensus

In the model of \cite{snyder2005reliability}, the failure of multiple facilities is considered. Each costumer is assigned a primary facility and a set of ordered backup facilities to be chosen if the prior facilities fail. Their formulation is solved using Lagrangian relaxation within a branch-and-bound method. In \cite{cui2010reliable}, the authors extend this model by allowing location specific rather than homogeneous failure probabilities. They propose a continuous approximation heuristic that omits detailed location and assignment decisions to obtain near-optimal solutions efficiently
In \cite{an2015reliable}, the authors propose a nonlinear mixed integer programming model that jointly determines backup hub selection and alternative routing to proactively address hub disruptions. Building on this, \cite{rostami2018reliable} introduce a two-stage formulation for the single allocation hub location problem that considers a set of possible hub breakdown scenarios. By applying a decomposition strategy, they preserve a structure similar to the classical single-allocation hub location problem, which enables efficient solution through linearization techniques.   \cite{tran2017reliable} assume that multiple hub can simultaneously fail. For each hub a sequence of backup hubs is determined, where a backup hub takes over all flow from the original hub. \cite{blanco2023hub} shift the focus from hubs to the reliability of interhub links, proposing models that construct protected networks with alternative routing paths in case of link failures. This emphasis on protecting the connections between hubs rather than the hubs themselves complements the node-focused strategies discussed in the other papers.

These studies underscore the critical importance of preemptively addressing failure scenarios within network design, offering valuable insights for designing optimization model under leader failures in parallel consensus network.  By focusing on the optimal selection of backup nodes for the leader, our work aims to enhance the efficiency and reliability of the view change protocol in parallel BFT consensus.

\section{Preliminaries} \label{sec:pre}

\subsection{Parallel Byzantine consensus network} \label{sec:parbyz}
In this paper, we develop an optimization model for the view change configuration of the parallel BFT consensus networks. 
Performance in parallel BFT consensus algorithm can be improved by the use of a hierarchical structure, splitting nodes into an upper layer and a lower layer \citep{amir2008steward,gupta2020resilientdb,thai2019hierarchical}. Introducing a hierarchical structure which consists of a verification committee above multiple consensus committees, separates global ordering from local transaction processing. The verification layer handles cross‐shard coordination and enforces the necessary total order on transactions among a set of nodes, while nodes in each consensus committee process their own transactions. This hierarchical structure has been shown to significantly improve scalability and performance compared to flat, single-committee consensus \citep{chen2024parallel}.
Based on these advancements, our research builds on the hierarchical sharded blockchains with the BFT intra-committee consensus protocol and seeks to improve performance further by optimizing the parallel committee configuration when view change occurs. 

We consider a node set in a parallel Byzantine consensus network. More specifically, we consider that  is divided as follows:
1) node set of a single verification committee, denoted $\mathcal{S}_v$ which is made up of $n_v$ reliable nodes, which perform the signature verification process. The nodes in $\mathcal{S}_v$ is selected among all the system network nodes based on their  reliability, network connectivity and capacity.
2) node set of $K$ parallel consensus committees, denoted as $\mathcal{N}$ which is made up of the remaining $n$ nodes in the system.

The node set of ${k}$th consensus committee is denoted $\mathcal{N}_k$, ${k} \in \mathcal{K}$ where $\mathcal{K}$ is the set of consensus committees.
We note $\mathcal{N} = \bigcup\limits_{k=1}^{k={K}} \mathcal{N}_k$. Thus, the whole node set is partitioned into $K+1$ disjoint committees  ($ \mathcal{S}_v  \bigcup\limits_{k=1}^{k={K}} \mathcal{N}_k $) and includes $ n_v + n$ nodes.

\subsection{View change protocol}

In Parallel BFT consensus, a view change is executed within each committee which can be triggered by one of the following conditions: 1) an expired timeout, indicating that no response has been received from the leader within a preset time; 2) a faulty operation detected by a follower within the leader's committee; or 3) a faulty operation identified by the verification committee, which will then broadcast a view-change request to the corresponding consensus committee \citep{castro1999practical,chen2023parbft,zhang2024prestigebft}. When any of these conditions are met, the leader node of view $v$ is considered failed, the algorithm proceeds to view $v+1$ and selects a new leader node. In the following, we describe the view-change procedure from the perspective of an individual committee, using the protocol in \citep{castro2002practical} as a representative example. 
As shown in Fig. \ref{fig:vc}, this process involves the committee reconfiguration of follower nodes from the detection of a faulty leader node to the new selected leader node.

\begin{figure}[h]
    \centering
    {
    \includegraphics[scale=0.35]{ 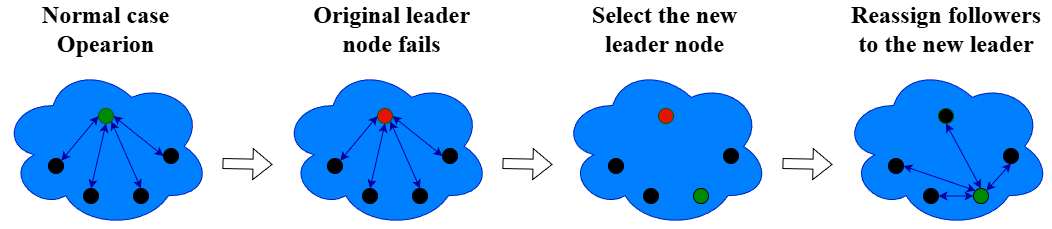} 
    \caption{View change protocol phases under leader node failure.}
    \label{fig:vc}}
\end{figure}

A view change process has three phases (shown in Fig. \ref{fig:vc-phases}). Upon confirming a leader failure based on any of the based on any of the aforementioned triggers, the follower node considers the leader node failed in view $v$ and prepares to enter view $v+1$ by starting the view change phase. It broadcasts a \textsc{view-change} message, where $\mathrm{newView} = v + 1$. After receiving a \textsc{view-change} message, the follower node sends a \textsc{view-change-ack} to the new leader node of view $v+1$. After receiving at least $\frac{2N-5}{3}$ \textsc{view-change-ack} messages that acknowledge follower's \textsc{vie-change} message, the new leader node consolidates the system state by reconciling checkpoints and carrying over uncommitted requests from the previous view. Finally, followers validate the correctness of this constructed state. If the new leader's proposal is deemed legitimate, the view-change protocol terminates and normal operation resumes. Otherwise, the system initiates a further leadership rotation and move to a new view $v+2$.

\begin{figure}[h]
    \centering
    {
    \includegraphics[scale=0.5]{ 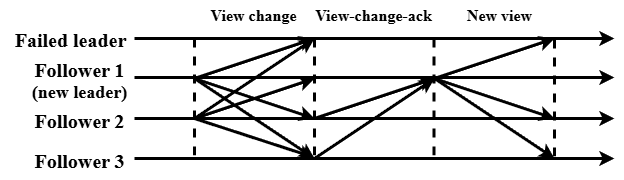} 
    \caption{View change protocol phases under leader node failure.}
    \label{fig:vc-phases}}
\end{figure}


\subsection{Problem definition and assumptions}

Throughput is commonly used as a performance metric to evaluate the performance of processing systems \citep{coccoli2002performance}.
Given an observation period, transaction throughput is computed as the ratio between the number of successfully committed transactions and the time period. Consequently, obtaining higher throughput is equivalent to successfully committing a higher number of transactions during the observation period. If we view the system from a single transaction point of view, the time spent to successfully commit this transaction should be as small as possible. 
Our goal is to minimize the transaction time by optimizing the configuration of the parallel consensus network.

Consider a graph $G=(\mathcal{N}, \mathcal{E})$, where the node set $\mathcal{N}$ represents a given set of the nodes in a parallel consensus network to be configured and the edge set $\mathcal{E}$ represents the potential communication links between pairs of nodes. Any pair of network nodes can be connected through potential communication links.
Let $d_{ij}>0$ denote the delay for  communication link $(i,j) \in \mathcal{E}$ between node $i$ and node $j$.

In the remainder of this paper we work under the following assumptions:
\begin{itemize}
    \item In constructing the optimization model, we disregard the time taken by each node to process requests, as this time is negligible compared to the communication delays. Thus in our optimization model we focus on  whether each pair of nodes is connected or not, and the communication times as these have most impact on the partitioning. 
    \item We focus on optimizing the long-term average performance of the consensus network. Since the duration of the normal operations typically substantially exceeds that of the transient view-change phases, the system's overall throughput is dominated by its performance in the normal case. Consequently, our optimization aims to minimize transaction operation latency within views, rather than the latency between views. In this model, the view-change duration is treated as a fixed protocol overhead and is therefore excluded from the objective function calculations.
    \item In a hierarchical parallel BFT consensus network, the upper-level verification committee typically requires significant computational power as it handles verification for multiple parallel consensus committees. Therefore, the nodes for the verification committee are usually designated rather than determined by our partition optimization model. It is important to note that our optimization model can still address the partitioning configuration issues of a flat-structured parallel consensus network by simply omitting the portion influenced by the verification committee in the objective function. 
\end{itemize}

\section{View change optimization model}

In this section, we detail our VCO model to achieve performance improvement of parallel BFT. The purpose of our optimization model is threefold:  \romannumeral1) the optimal committee configuration which decides if nodes are connected or not under the normal case replication protocol, \romannumeral2) identify the nodes which will serve as leaders of these committees, and \romannumeral3) optimally decide the backup node for the current leader node in each committee.

\begin{table}[!htb]
\centering
\caption{Notation list of the VCO model.} \label{tab: notations}
\begin{tabular}{l|l}
\cline{1-2}
Notation  & Description  \\
\cline{1-2}
 $\mathcal{N}$ & Set of nodes to be configured in consensus committees: $n$ nodes. \\
 $\mathcal{S}_{v}$ & Set of nodes in the verification committee.\\
 $d_{ij}$ & The communication delay from node $i$ to node $j$.\\
 $d_{iv}$/$d_{vi}$ & The communication delays from  the consensus committee led by   \\
             & node $i$ to the verification committee and vice versa, respectively. \\    
 $f_i$ & Failure probability of node $i$. \\
 $f_{\min}$ & Minimum number of tolerated faulty nodes per consensus \\ & committee.\\
 $x_{ij}$ \small{$({i,j} \in \mathcal{N})$} \; & A binary decision variable equal to 1 if node $i$ is the leader of  \\
 &  node $j$, equal to 0 if not.  \\
 $y_i^k$ \small{$({i,j} \in \mathcal{N})$} & A binary decision variable equal to 1 if node $k$ is a backup node \\
 & for leader $i$, equal to 0 if not. \\
 $z_{ij}^k$ \small{$({i,j} \in \mathcal{N})$} & A binary decision variable equal to 1 if  node $j$ is reassigned \\ & from leader node $i$ to 
  backup leader node $k$, equal to 0 if not. \\
\cline{1-2}
\end{tabular}
\end{table}

\subsection{Decision variables} \label{sec:decvar}

We introduce decision variables to configure the committee structure optimally. These variables are of two types. Firstly, corresponding to the normal case protocol, we define variables to determine which nodes are selected to serve as leader nodes for the follower nodes.
\begin{align*}
&x_{ij}=
\biggl\{
    \begin{aligned}
    1 \quad & \text{if node $i$ serves as a leader node of follower node $j$} \\
    0 \quad & \text{otherwise}
    \end{aligned}
\end{align*}
In the case $i=j$, $x_{ii}=1$ means $i$ is a leader and $x_{ii}=0$ means it is a follower.

These decision variables address three key aspects of the consensus network configuration:
\begin{enumerate}[label=\arabic*)]
\item Total number of committees: The total number of committees is determined by the number of leader nodes. Since each committee has a unique leader, the number of committees is represented by the sum of binary variables $p = \sum_i x_{ii}$. 

\item Selection of the leader nodes for each consensus committee: This involves identifying which nodes will act as leaders within the network. The leader nodes are responsible for coordinating communication with the verification committee and communicating with follower nodes in the committee.

\item Determination of which follower nodes are allocated to the leader nodes: Once the leader nodes are selected, the remaining nodes will be allocated as followers to these leaders. 
\end{enumerate}
During the view change, we introduce decision variables to determine the backup node for each leader node within their respective committees. In the event of a leader node failure, the backup node will be selected and coordinate the communication with other follower nodes and the verification committee.

\subsection{Objective function}  \label{sec:objfun}

Under both the normal case replication protocol and the view-change protocol, our objective is to minimize latency. Within the configuration of the consensus committee, the communication delay within the committee and the delay between the committee leader and the verification committee constitute the primary components of transaction delay. For any nodes $i$ and $j$, we define $d_{ij}$ as the delay between node $i$ and node $j$, and $d_{iv}$ as the delay between node $i$ and the verification committee. The approach of \cite{chen2023parbft} minimizes the latency of one successful transaction. In their formulation, they minimize the total latency for each phase; in each phase the latencies are computed by the the longest delay for participating nodes. Their work focuses on optimizing the performance of their specific protocol so they minimize the latency as the sum of each phase. In our model, we aim to optimize the performance under all the types of parallel BFT protocols. Therefore, instead of explicitly expressing the objective function as the sum of delays for each phase, we directly minimize the total delay between connected nodes within the consensus committee and between the consensus leader and the verification committee. Thus we have: 

\begin{align*}
    \min & \quad \sum_{i,j\in \mathcal{N}} d_{ij} x_{ij} + \sum_{i \in \mathcal{N}} d_{iv} x_{ii}
\end{align*}


\subsection{Optimization under normal case}
\label{sec:onc}
The normal-case parallel BFT consensus optimization problem involves selecting leader nodes for committees and assigning the remaining follower nodes to these leaders.
The optimization problem under normal case can be formulated as:

\begin{align}
    \min & \quad \sum_{i,j\in \mathcal{N}} d_{ij} x_{ij} + \sum_{i \in \mathcal{N}} d_{iv} x_{ii} \label{nc-obj1} \\
    \text{s.t.}
    & \quad j \in \mathcal{N}: \  \sum_{i=1}^{n} x_{ij} = 1 \label{nc-constr1} \\
    & \quad i,j \in \mathcal{N}: \  x_{ij} \leq x_{ii} \label{nc-constr2} \\
    & \quad i,j \in \mathcal{N}: \  \sum_{j=1}^{n} x_{ij} + 1 \geq 3 f_{\min} + 1 \label{nc-constr3} \\
    & \quad i,j \in \mathcal{N}: \  x_{ij} = \{0,1\} \label{nc-constr4}
\end{align}

Constraints (\ref{nc-constr1}) ensure that each follower node is assigned to exactly one leader node. This is crucial that every node has one and only one designated leader, preventing any follower from being unassigned or assigned to multiple leaders. 

Constraints (\ref{nc-constr2}) ensure that a node can only be selected as a leader if it is assigned to itself. The left side of constraint~(\ref{nc-constr2}) can take the value of 1 only when $x_{ii} = 1$. 
When $x_{ii} =0$, the node $i$ serves as a follower of a consensus committee, hence it cannot be a leader. Constraint~(\ref{nc-constr2}) prevents any follower from communicating with nodes who are not their leader (i.e.\ leaders and followers of other committees).

Constraints (\ref{nc-constr3}) ensure that each leader node has at least \( 3 f_{\min}  \) followers. This constraint is derived from the principles of Byzantine Fault Tolerance, which require a certain number of nodes to achieve consensus despite potential faults or malicious actors. By maintaining this minimum number of followers, the network can withstand up to \( f_{\min} \) faulty nodes, thus ensuring the BFT safety requirements.

\subsection{Optimization model for view change}
\label{sec:VCOP}
The view change protocol addresses the problem of a faulty leader node to ensure the \emph{liveness} property, which guarantees that the system continues to process transactions and reach consensus despite any leader failures. In the context of parallel BFT, which employs multiple committees to achieve parallel consensus, it is essential to optimize the view change mechanism for this parallel network. We separately consider the scenario where the leader node of each parallel committee encounters a failure. The probability of a view change for a specific consensus committee is equivalent to the failure rate of its corresponding leader node.

To address leader node failures, we define a finite set \( S \) representing all possible scenarios of leader node failure. Since the simultaneous failure of several leader nodes is highly improbable, we assume that at most one leader node fails at a time. This leads to a finite set \( S = \{S_0, S_1, \ldots, S_n\} \), which consists of \( n+1 \) scenarios. Here, \( S_0 \) represents the normal case scenario where all leader nodes are operational, while \( S_1, \ldots, S_n \) describe scenarios where each specific leader node fails individually. The probability of each scenario \( S_i \) occurring is denoted by \( p_{S_i} \) for \( i \geq 1 \).

For each realized failure scenario, a preselected backup leader node replaces the failed leader node, meaning the corresponding follower nodes are re-allocated to the backup leader node. Then we need to make two types of decisions. First, in the normal scenario \( S_0 \), a subset of nodes is selected as leader nodes, and all other nodes are allocated as followers to exactly one leader node. 
Second, for each failure scenario \( \{S_i \mid 1 \leq i \leq N \} \), a response must be formulated to address the failure of leader node \( i \). This involves selecting a backup leader node and reassigning all follower nodes of the failed leader node to the backup leader node. 


In order to model the view change decision, we define the backup nodes selection and reassignment decision variables in the following:
\begin{align*}
&y_i^k=
\begin{cases}
    1 & \text{if node $k$ is a backup node for leader $i$}, \\
    0 & \text{otherwise}.
\end{cases}\\
&z_{ij}^k=
\begin{cases}
    1 & \text{if follower node $j$ is reassigned from leader node $i$ to backup node $k$}, \\
    0 & \text{otherwise}.
\end{cases}
\end{align*}
For a given failure scenario \( S_i \), we introduce two network latency functions: \( f_{S_i} \) and \( g_{S_i} \). The function \( f_{S_i}^k \) represents the additional transaction processing latency that arises due to the view change, specifically the reassignment to the leader node \( k \) following a failure. This additional latency is computed as the sum of the network delay between the backup leader node \( k \) and its follower nodes \( j \), as well as the network delay between the backup leader node \( k \) and the verification committee. Mathematically, this can be expressed as:
\begin{align}
     f_{S_i}^k = d_{kv} y_i^k + \sum_{j \in \mathcal{N}} d_{kj} z_{ij}^k \label{vc-obj-f}
\end{align}

Conversely, some transaction delays are mitigated in scenario \( S_i \) due to the failure of leader node \( i \). These reductions in delay consist of the transaction delays that would have occurred between the failed leader node \( i \) and its follower nodes, as well as between the failed leader node \( i \) and the verification committee. This can be formulated as:
\begin{align}
      g_{S_i} = d_{iv} x_{ii} + \sum_{j \in \mathcal{N}} d_{ij} x_{ij} \label{vc-obj-g}
\end{align}

Here, \( g_{S_i} \) captures the saved latency, or rather the reduction  in network delay, that results from the failure of leader node \( i \). Specifically, \( d_{iv} x_{ii} \) represents the delay reduction between the failed leader node \( i \) and the verification committee, while \( \sum_{j} d_{ij} x_{ij} \) captures the delay reduction between the failed leader node \( i \) and its follower nodes. Due to the leader node switch, we are calculating the difference in total delay for the failure case scenario. Therefore, we need to subtract the original delays between the former leader node and the verification committee, as well as with the follower nodes. 

The overall latency for processing transactions after a view change can then be formulated by integrating the probabilities of each failure scenario. This provides a weighted sum of the additional latencies incurred and the latencies mitigated due to the failure of leader nodes. The formula is expressed as:
\begin{align}
    Q(x) = \sum_{i=1}^n p_{S_i} Q_{S_i}(x) = \sum_{i=1}^n p_{S_i} \Big( \sum_{k \in \mathcal{N}} f_{S_i}^k - g_{S_i} \Big) \label{vc-obj-q}
\end{align}

Here, \( p_{S_i} \) represents the probability of scenario \( S_i \) occurring. The term \( Q_{S_i}(x) \) encapsulates the net latency for scenario \( S_i \), which is the difference between the additional latency \( \sum_{k} f_{S_i}^k \) due to reassignment of leader nodes and the latency reduction \( g_{S_i} \) due to the failure of the original leader node. By summing these weighted net latencies across all possible scenarios, we obtain \( Q(x) \) as the expected latency for opeations after a view change.
\begin{align} 
    \min & \quad \sum_{i \in \mathcal{N}} d_{iv} x_{ii} + \sum_{i,j \in \mathcal{N}} d_{ij} x_{ij} + Q(x) \label{vc-obj}\\
    \text{s.t.} & \quad i \in \mathcal{N}: \  \sum_{k \in \mathcal{N}, k \neq i} y_i^k = x_{ii} \label{vc-constr1}\\
    & \quad i,k \in \mathcal{N}: \  y_i^k \leq x_{ik} \label{vc-constr2}\\
    & \quad i,j,k \in \mathcal{N}: \  y_i^k \leq z_{ij}^k \label{vc-constr3}\\
    & \quad i,j \in \mathcal{N}: \  \sum_{k \in \mathcal{N}, k \neq i} z_{ij}^k = x_{ij} \label{vc-constr4}\\
    & \quad j \in \mathcal{N}: \  \sum_{i=1}^{n} x_{ij} = 1 \label{vc-constr5}\\
    & \quad i,j \in \mathcal{N}: \  x_{ij} \leq x_{ii} \label{vc-constr6}\\
    & \quad i,j \in \mathcal{N}: \  \sum_{j=1}^{n} x_{ij} + 1 \geq 3 f_{\min} + 1 \label{vc-constr7}\\
    & \quad i,j,k \in \mathcal{N}: \  x_{ij}, y_i^k, z_{ij}^k \in \{0,1\} \label{vc-constr8}
\end{align}

The objective function (\ref{vc-obj}) minimizes the total network delay. This objective function includes three main components: the sum of the delays between leader nodes and the verification committee (\( \sum_{i \in \mathcal{N}} d_{iv} x_{ii} \)), the delays between leader nodes and their follower nodes (\( \sum_{i,j \in \mathcal{N}} d_{ij} x_{ij} \)), and the expected network latency after a view change (\( Q(x) \)). 

Constraints (\ref{vc-constr1}) ensure that for each failed leader node \(i\), exactly one backup node is chosen. Here, \( y_i^k \) is a binary variable indicating whether node \( k \) is the backup for leader node \( i \). Constraints (\ref{vc-constr1}) guarantees that every leader node \( i \) has a designated backup.

Constraints (\ref{vc-constr2}) enforce that a node \( k \) can only serve as a backup for leader node \( i \) if \( i \) is selected as the leader node of the backup node before the view change. 

Constraints (\ref{vc-constr3}) ensure that a follower node \( j \) can be reallocated to the backup node \( k \) only if node \( k \) is already designated as a backup node for leader node \( i \). They ensure that follower nodes are only reassigned to valid backup nodes, thereby preventing misallocation. 

Constraints (\ref{vc-constr4}) ensure that for every allocation \( x_{ij} \) to the failed leader node \( i \), a corresponding reallocation \( z_{ij}^k \) is chosen during the view change. These constraints ensure that all follower nodes originally assigned to a failed leader node are properly reassigned to backup nodes. 

Constraints (\ref{vc-constr5}) to (\ref{vc-constr7}) are similar to those in the normal case and ensure that each follower node is assigned to exactly one leader node, that a follower node can only be assigned to a valid leader node, and that the system maintains sufficient nodes to tolerate up to \( f_{\min} \) failures, respectively. 

\subsection{Backup node selection as views proceed} \label{sec:bsvp}

The view change process plays an essential role in keeping a parallel BFT system operational when multiple leader failures occur over time. In this subsection, we describe how backup leaders are selected as views proceed and explain how the system transitions from one leader to the next after the view change reconfiguration. The key idea is that once the initial committee configuration \(x\) has been fixed, subsequent leader replacements do not require re-optimizing \(x\). 
Building on the optimization models introduced in Section~\ref{sec:VCOP}, we design an iterative procedure that selects replacement leaders while seeking to minimize the latency added by each view change.

The backup leader selection process operates as views proceed as follows:

\begin{itemize}
    \item[1.] \textbf{View initialization.} The system initializes the decision variables $x$ obtained by solving the view change model. For each view, we then need to optimally obtain the decision variables \(y_i^k\) and the follower-reassignment variables \(z_{ij}^k\). These optimal solution variables will indicate which node should replace a failed leader and how its followers should be reassigned.
    \item[2.] \textbf{Backup node selection by solving VCO model.} The system considers each possible leader-failure scenario \(S_i (i \geq 1, S_i \in S)\) for the node $i$ which $x_{ii}=1$. For a scenario in which leader \(i\) fails, the algorithm solves the model (\ref{eq:rvco-obj})–(\ref{eq:rvco-constr5}) to choose a new backup node \(k\) for leader \(i\). 
    \begin{align}
\min & \quad \sum_{k \in \mathcal{N}} f_{S_i}^k - g_{S_i} \label{eq:rvco-obj}\\
\text{s.t.} & \quad \sum_{k \in \mathcal{N}, k \neq i} y_i^k = x_{ii} \label{eq:rvco-constr1}\\
& \quad  k \in \mathcal{N}, k \neq i: \  y_i^k \le x_{ik} \label{eq:rvco-constr2}\\
& \quad j \in \mathcal{N}: \  \sum_{k \in \mathcal{N}, k \neq i} z_{ij}^k = x_{ij} \label{eq:rvco-constr3}\\
&\quad j, k \in \mathcal{N}, k \neq i: \  y_i^k \le z_{ij}^k  \label{eq:rvco-constr4}\\
& \quad j,k \in \mathcal{N}: \  y_i^k, z_{ij}^k \in \{0,1\} \label{eq:rvco-constr5}
\end{align}
    \item[3.] \textbf{Follower reassignment if the current leader fails.} Once the current leader node \(i\) is identified, a view change is driven and a backup node \(k\) has been replaced with node \(i\) for the next view, all of node \(i\)’s followers are reassigned. This step is determined by optimal decision variable \(z_{ij}^k = 1\) for every follower \(j\) of the model (\ref{eq:rvco-obj})–(\ref{eq:rvco-constr5}). 
    \item[4.] \textbf{Node state update.} The system updates the leader–follower configuration. The selected backup leader \(k\) becomes the active leader, and the configuration is updated by setting \(x_{kk}=1\), clearing the failed leader’s status \(x_{ii}=0\), and updating the follower relationships accordingly. For later failures, the committee configuration \(x\) remains unchanged, and only the backup-selection variables are re-optimized. 
\end{itemize}

This algorithm supports dynamic and efficient backup leader selection as views proceed, which ensures low latency in parallel BFT networks. The corresponding pseudocode is shown below. 

\begin{algorithm}[h]
\caption{Backup leader selection as views proceed}
\label{alg:backup}
\begin{algorithmic}[1]
\For{each view $t$}
    \State Initialize the leader--follower configuration for view $t$:
    \State Set $x_{ii}=1$ and $x_{ij}=1$ if node $i$ is the leader and node $j$ is its follower.
    \Statex \hspace{1.4em}Otherwise set $x_{ii}=0$ and $x_{ij}=0$.
    \For{each leader failure scenario $S_i$}
        \State Identify leader $i$ and its associated followers.
        \State Solve model (\ref{eq:rvco-obj})--(\ref{eq:rvco-constr5}) and obtain the optimal $y_i^k, z_{ij}^k$.
    \EndFor
    \If{leader node $i$ fails in view $t$}
        \For{each follower $j$ of the failed leader $i$}
            \State Reassign follower $j$ to the backup leader $k$ if $z_{ij}^k = 1$.
        \EndFor
        \Statex \hspace{3em}Update system state:
        \Statex \hspace{3em}The backup leader $k$ becomes the new leader.
        \Statex \hspace{3em}Set $x_{ii} \gets 0$, $x_{kk} \gets 1$.
        \Statex \hspace{3em}For each follower $j$, set $x_{kj} \gets 1$ and $x_{ij} \gets 0$.
    \EndIf
    \State Proceed to next view: $t \gets t+1$
\EndFor
\end{algorithmic}
\end{algorithm}

\section{Solution method for view change optimization model}

In this section, we apply a decomposition method to solve the view change optimization problem within parallel BFT network. Our approach leverages Benders decomposition, breaking down the view change network optimization into a master problem and a collection of sub-problems, each corresponding to a distinct scenario \( S_i \), where \( i \geq 1 \). This decomposition is suitable because the subproblems associated with each scenario \( S_i \) (for \( i \geq 1 \)) can be solved independently, given a fixed set of master problem decision variables. 

In the master problem, the decision variables \( x_{ij} \)  indicate the initial configuration of the network, which in turn influences view change process. Once an optimal set of decision variables \( \hat{x} \) is identified through the master problem, the focus shifts to solving the subproblems. Each of these subproblems is solved by determining the optimal backup leader node and reallocation of follower nodes, as captured by the backup variables \( y \) and reallocation variables \( z \). The subproblems are solved independently for each scenario, which allows for a response to different potential fault scenarios.

The master problem is formalized as follows:
\begin{align}
    MP: \quad 
   \min & \quad  \sum_{i \in \mathcal{N}} d_{iv} x_{ii} + \sum_{i,j \in \mathcal{N}}d_{ij} x_{ij} + \theta \label{mp-obj} \\
   \text{s.t.} & \quad  \theta \geq Q(x) \label{mp-constr1} \\
   & \quad (\ref{nc-constr1}) - (\ref{nc-constr4})  \label{mp-constr2}
\end{align}
where \(\theta\) represents the total reconfiguration delay across all subproblems.

Let \( X \) denote the set of all binary vectors associated with the \( x_{ij} \) variables. For a given vector \(\hat{x} \in X\) and the current value of \( \theta \), denoted as \(\hat{\theta}\), at a node in the enumeration tree within the branch-and-bound method for the master problem, and for a given scenario \( S_i \) where \( i \geq 1 \), we consider the following primal subproblem:
\begin{align}
    PS(\hat{x}, S_i): \quad \min &  \quad  Q_{S_i}(\hat{x}) = \sum_k f_{S_i}^k - g_{S_i}  \label{ps-obj} \\
    \text{s.t.} &  \quad \sum_{k \in \mathcal{N}} y_i^k = \hat{x}_{ii} \label{ps-constr1} \\
    &  \quad k \in \mathcal{N}: \  y_i^k \leq \hat{x}_{ik} \label{ps-constr2} \\
    &  \quad j,k \in \mathcal{N}: \  y_i^k \leq z_{ij}^k \label{ps-constr3} \\
    &  \quad  j \in \mathcal{N}: \  \sum_{k \in \mathcal{N}} z_{ij}^k = \hat{x}_{ij} \label{ps-constr4}  \\
    &  \quad  j,k \in \mathcal{N}: \  y_i^k, z_{ij}^k \in \{0,1\}  \label{ps-constr5}
\end{align}

The constraints in this primal subproblem arise from the view change network optimization problem with the allocation vector \( \hat{x} \) fixed. The primal subproblem \( PS(\hat{x},S_i) \) thus identifies the optimal backup leader node to assume the role of the failed leader node \( i \). Note that the feasibility of \( PS(\hat{x},S_i) \) is assured by the constraints (\ref{mp-constr2}) in the master problem. The objective function \( Q_{S_i}(\hat{x}) \) represents the additional delay introduced by reconfiguring all followers from the failed leader node \( i \) to the newly designated backup leader node.

The optimal value of each subproblem, denoted \( Q^*_{S_i}(\hat{x}) \), represents the minimized delay for scenario \( S_i \) given the fixed normal case allocation decision variables \(\hat{x}\). The overall optimal value of view change problem combining all scenarios is then formulated as:
\begin{align}
    Q^*(\hat{x}) = \sum_{i=1}^n p_{S_i} Q_{S_i}^*(\hat{x}) \label{sp-optv}
\end{align}


\subsection{Efficient solution of the subproblems} \label{sec:essp}

To efficiently solve the subproblem \( PS(\hat{x}, S_i) \), we consider an iteration of Benders' decomposition algorithm at a given node of the enumeration tree, with the binary solution \( \hat{x} \). Let \(H = \{ i \in \mathcal{N} : \hat{x}_{ii} = 1 \}\) denote the set of nodes selected as leaders under optimal configuration \(\hat{x}\). In our formulation, each scenario \(S_i\) corresponds to the failure of a specific node \(i \in N\). For scenarios where \(i \notin H\), a failure of node \(i\) does not trigger a view change because \(i\) is not a designated leader. Consequently, the subproblem \(PS(\hat{x}, S_i)\) reduces to a trivial case whose optimal value is zero, since no leader-replacement decisions are required. This observation allows us to only consider subproblems where \( i \in H \).

In cases where \( i \in H \), we can leverage the structure of the subproblem to efficiently determine its optimal value \( Q_{S_i}^*(\hat{x}) \). Define \( N_{S_i} = \{j \in \mathcal{N} : \hat{x}_{ij} = 1\} \) as the set of all follower nodes assigned to a failed leader node \( i \). With this definition, the subproblem \( PS(\hat{x}, S_i) \) can be reformulated as:
\begin{align}
    PS'(\hat{x}, S_i): \quad \min &  \quad  Q_{S_i}(\hat{x}) = \sum_{k \in N_{S_i}} \Big(d_{kv} y_i^k + \sum_{j \in N_{S_i}} d_{kj}z_{ij}^k\Big) - g_{S_i}(\hat{x}) \label{psvar-obj}\\
     \text{s.t.} &  \quad \sum_{k \in N_{S_i}} y_i^k = 1 \label{psvar-constr1} \\
    &  \quad j,k \in N_{S_i}: \ y_i^k \leq z_{ij}^k \label{psvar-constr2} \\
    & \quad j \in N_{S_i}: \ \sum_{k \in N_{S_i}} z_{ij}^k = 1 \label{psvar-constr3} \\
    & \quad  j,k \in N_{S_i}: \  y_i^k, z_{ij}^k \in \{0,1\}  \label{psvar-constr4} 
\end{align}

This reformulation shows that the subproblem \( PS'(\hat{x}, S_i) \) is a variant of the classic 1-median problem \citep{daskin2015p}. The 1-median problem seeks to identify the optimal location for one facility (in this case, the backup leader node) that minimizes the total distance of serving a set of demand points (the follower nodes). It can be solved in an efficient way as explained below.

The efficient solution of this subproblem evaluates the additional network delay incurred for each potential backup node \( k \in N_{S_i} \). This delay is computed as:
\begin{align}
    exd(k) = d_{kv} y_i^k + \sum_{j \in N_{S_i}} d_{kj}z_{ij}^k \label{exobj}
\end{align}

Here, \( exd(k) \) represents the extra delay introduced when node \( k \) is selected as the backup leader for the failed leader node \( i \). The objective is to find the backup node \( k \in N_{S_i} \) that minimizes this additional delay. 
To solve this, we compute \( exd(k) \) for each candidate backup node \( k \) in the set \( N_{S_i} \). The optimal solution is then obtained by selecting the node that yields the minimum value of \( exd(k) \).

\subsection{Improved Benders cuts} \label{sec:ibc}

To enhance the efficiency of Benders decomposition in solving the VCO problem, we propose a modification that leverages the dual information derived from the Linear Programming (LP) relaxation of the subproblems. 
We begin by revisiting the subproblem \( PS(\hat{x}, S_i) \). 
A key observation is that the subproblem \( PS(\hat{x}, S_i) \) possesses an integrality property under certain conditions, allowing us to work with the LP relaxation. This is formalized in the following proposition:

\begin{proposition} \label{integrality}
    For any binary \( \hat{x} \) and any \( i \geq 1 \), the subproblem \( PS(\hat{x}, S_i) \) has the integrality property.
\end{proposition}

\begin{proof}
Consider the LP relaxation of the reformulated subproblem \( PS(\hat{x}, S_i) \) as defined by equations (\ref{psvar-obj}) to (\ref{psvar-constr3}). By subtracting the constraint (\ref{psvar-constr1}) from (\ref{psvar-constr3}), we obtain:
\begin{align}
    j \in \mathcal{N}: \ \sum_{k \in N_{S_i}} \big(y_i^k - z_{ij}^k\big) = 0 \label{psvar-constr5}
\end{align}

This equality (\ref{psvar-constr5}) together with constraint (\ref{psvar-constr2}) implies that:
\begin{align}
    k \in N_{S_i}, j \in \mathcal{N}: \  y_i^k = z_{ij}^k \label{yequalz}
\end{align}

This equality simplifies the LP relaxation of the subproblem to the following:
\begin{align}
    \min & \quad Q_{S_i}(\hat{x}) = \sum_{k \in N_{S_i}} \Big(d_{kv} y_i^k + \sum_{j \in N_{S_i}} d_{kj}y_i^k\Big) - g_{S_i}(\hat{x}) \label{psr-obj} \\
    \text{s.t.} & \quad  \sum_{k \in N_{S_i}} y_i^k = 1 \label{psr-constr1} \\
    & \quad k \in N_{S_i}: \  0 \leq y_i^k \leq 1 \label{psr-constr2}
\end{align}

This LP relaxation retains the integrality property, meaning that the optimal solution to the LP relaxation is identical to that of the original integer programming problem. 
\end{proof}

Thus, we can remove the integrality requirement on the variables in $PS(\hat{x},S_i)$ and form its dual linear programming problem. For this, let $\alpha$, $\beta$, $\gamma$ and $\lambda$ be the dual variables corresponding to constraints (\ref{ps-constr1}) to (\ref{ps-constr4}), respectively. The dual of $PS(\hat{x}, S_i)$ can then be stated as 
\begin{align}
    DS(\hat{x}, S_i): \quad \max &  \quad \alpha \hat{x}_{ii} + \sum_{k\neq i} \beta_k \hat{x}_{ik} + \sum_j \lambda_j \hat{x}_{ij} - g_{S_i}(\hat{x}) \label{ds-obj}\\
     \text{s.t.} & \quad  k \in \mathcal{N}: \  \alpha + \beta_k + \sum_j \gamma_{kj} \leq d_{kv}, \quad \forall k. \label{ds-constr1}\\
    &  \quad k,j \in \mathcal{N}: \  \lambda_j - \gamma_{kj} \leq d_{kj}. \label{ds-constr2} \\
    & \quad \beta, \gamma \leq 0. \label{ds-constr3}
\end{align}

By duality, one optimal solution for $DS(\hat{x}, S_i)$ is $(\alpha, \beta,\gamma,\lambda)=(\alpha,0,0,\lambda)$ with $\alpha=d_{kv}$ and $\lambda_j=d_{kj}$ for all $j$, where $k$ is the chosen backup node for leader node $i$. Then Benders cut is obtained as
\begin{align}
    \theta \geq \sum_{i=1}^n p_{S_i}\Big(\alpha x_{ii} + \sum_j \lambda_j x_{ij} - g_{S_i}(x)\Big) \geq \sum_{i=1}^n p_{S_i}\Big(d_{kv} x_{ii} + \sum_j d_{kj} x_{ij} - g_{S_i}(x)\Big) \label{newcut}
\end{align}

\subsection{Efficient backup node selection as views proceed}

In Section~\ref{sec:bsvp}, the view change process is handled by repeatedly solving (\ref{eq:rvco-obj})–(\ref{eq:rvco-constr5}) whenever a leader failure occurs. Although this approach guarantees optimal reconfiguration decisions, it incurs significant computational overhead when views proceed, particularly in parallel BFT systems where multiple committees may experience successive leader failures. Re-solving a MIP optimization problem at each view change is costly and limits the system’s ability to react promptly to failures.

Leveraging the results identified in Sections~\ref{sec:essp} and~\ref{sec:ibc}, we now present an efficient backup node selection algorithm that operates as views proceed. The key observation is that once the normal-case committee configuration \(x\) has been fixed by the master problem, subsequent view changes do not require re-optimizing \(x\). For any leader failure scenario \(S_i\), the corresponding subproblem reduces to a 1-median problem over the follower node set \(N_{S_i}\) of the corresponding failed leader node \(i\). These properties allow the backup leader to be selected through minimizing \(exd(k)\) rather than by solving a MIP. 
Also, equality (\ref{yequalz}) in the Proposition \ref{integrality} implies the follower reassignment variables $z_{ij}^k$ are determined by the leader-selection variables $y_i^k$ for follower nodes \(j\) in the node set \(N_{S_i}\). Consequently, once a backup leader $k$ is chosen (i.e., $y_i^k = 1$), the values of all $z_{ij}^k$ are fixed. Therefore, the reassignment does not require solving a MIP (\ref{eq:rvco-obj})–(\ref{eq:rvco-constr5}) to determine the $z$ variables.

Specifically, when a leader \(i\) fails in the current view, only nodes in the follower set
\(N_{S_i} = \{ j \in \mathcal{N} : x_{ij} = 1 \}\) are eligible to serve as replacement leaders. For each candidate node \(k \in N_{S_i}\), the additional delay incurred if \(k\) is selected as leader is given by
\[
exd(k) = d_{kv} + \sum_{j \in N_{S_i}} d_{kj}.
\]
The backup leader is then chosen as the node \(k^* = \arg\min_{k \in N_{S_i}} exd(k)\). Once \(k^*\) is selected, all followers of the failed leader \(i\) are reassigned to \(k^*\), and the leader--follower configuration is updated accordingly. For later view changes, the same efficient selection procedure given the updated configuration \(x\) is applied without revisiting the master problem.

The efficient backup leader selection procedure as views proceed is summarized in Algorithm~\ref{alg:efficient-backup}.

\begin{algorithm}[h]
\caption{Efficient backup leader selection algorithm as views proceed}
\label{alg:efficient-backup}
\begin{algorithmic}[1]
\For{each view $t$}
    \State Initialize the leader--follower configuration for view $t$:
    \State Set $x_{ii}=1$ and $x_{ij}=1$ if node $i$ is the leader and node $j$ is its follower.
    \Statex \hspace{1.4em}Otherwise set $x_{ii}=0$ and $x_{ij}=0$.
    \If{leader $i$ fails}
        \State Identify the follower set $N_{S_i} = \{ j : x_{ij} = 1 \}$.
        \For{each candidate $k \in N_{S_i}$}
            \State Compute $exd(k) = d_{kv} + \sum_{j \in N_{S_i}} d_{kj}$.
        \EndFor
        \State Select backup leader $k^* = \arg\min_{k \in N_{S_i}} exd(k)$.
        \State Select $k^*$ as the new leader: set $x_{k^*k^*} \gets 1$, $x_{ii} \gets 0$.
        \For{each follower $j \in N_{S_i}$}
            \State Reassign follower $j$ to $k^*$: set $x_{k^*j} \gets 1$, $x_{ij} \gets 0$.
        \EndFor
    \EndIf
    \State Proceed to the next view $t \gets t+1$.
\EndFor
\end{algorithmic}
\end{algorithm}

\section{Performance evaluation}
\label{sec:perform-eval}
We aim to show the performance enhancement achieved by our proposed VCO-driven parallel BFT algorithm and assess its robustness. 
To achieve rigorous evaluation results, we evaluate our algorithms in a real-world Cloud network environment (MS Azure VMs).
We extended the Java language testbed developed in \citep{chen2023parbft} to support our proposed parallel consensus architectures. Each node in the testbed can switch between consensus node and verification node according to the committee configuration scheme that specifies the assigned committee and the role of each node. 
In terms of performance measures, we measure the latency (in milliseconds, ms) as the time spent from sending a group of client requests (i.e.\ transactions) to accepting valid replies for all requests in the group at the client. We also measure the throughput as the number of operations processed every second in the system, represented by transactions per second, e.g.\ $op/s$.

\subsection{Experiment setup}

To benchmark the performance of VCO, we compare it with ParBFT \citep{chen2023parbft}, which is based on a BL-MILP optimization model, and non-optimized ParBFT\@. 
Based on ParBFT's initial configuration, we set the block size to 1 MB, and the message payload to 250 bytes in our experiments.
This is a common set-up in blockchain applications \citep{kotla2007zyzzyva,yin2019hotstuff,gueta2019sbft}. 
We set up five monitor cloud VMs on Microsoft Azure for running the ParBFT testbed: 1) we deploy nodes in the verification committee in one of the monitor cloud VMs, and 2) we deploy all the remaining nodes in consensus committees in the four remaining monitor cloud VMs. All VMs in the evaluation test use the same hardware: 8 v-CPUs and 32G RAM. 

We assume that each committee can tolerate $f_{\min}=1$ faulty node, which means the minimum number of nodes allocated in each committee should be no less than 4 nodes. We vary the total number of nodes in the consensus committee node set $\mathcal{N}$ from 40 to 200. The configuration of ParBFT is optimized by a BL-MILP model.
In the set-up of the non-optimized ParBFT,  we use the default committee partitioning scheme (i.e.\  at least four nodes randomly allocated to each committee), in which each committee has no more than $f_{\min}=1$ faulty node.
Provided that the minimum security conditions are satisfied, the number of committees and node allocation are set randomly in the non-optimized ParBFT.

\subsection{Performance comparison of parallel BFT}

In this section, we compare several BFT protocols under normal case operations. We compare the VCO-driven ParBFT algorithm to the baseline ParBFT configuration that uses a random committee organization scheme. This baseline serves as a reference for evaluating the impact of our optimization approach. The random organization scheme refers to a typical non-optimized setup in which committee formation and leader election follows no specific strategy.
we also compare VCO-driven ParBFT with HiBFT \citep{thai2019hierarchical}, a parallel Byzantine consensus protocol designed with a hierarchical structure. 
In addition, we evaluate the performance of VCO-driven ParBFT against ParBFT optimized by the BL-MILP model. Comparison are also conducted with HotStuff \citep{yin2019hotstuff}, a widely adopted benchmark protocol that employs a single-committee design and does not incorporate parallelization. We also include PrestigeBFT \citep{zhang2024prestigebft}, which incorporate reputation mechanism for leader selection and active view change for single-committee BFT algorithms.

\begin{figure*}[!htb]
    \centering
    \begin{subfigure}[b]{0.48\textwidth}
        \centering
        \includegraphics[width=\textwidth]{ 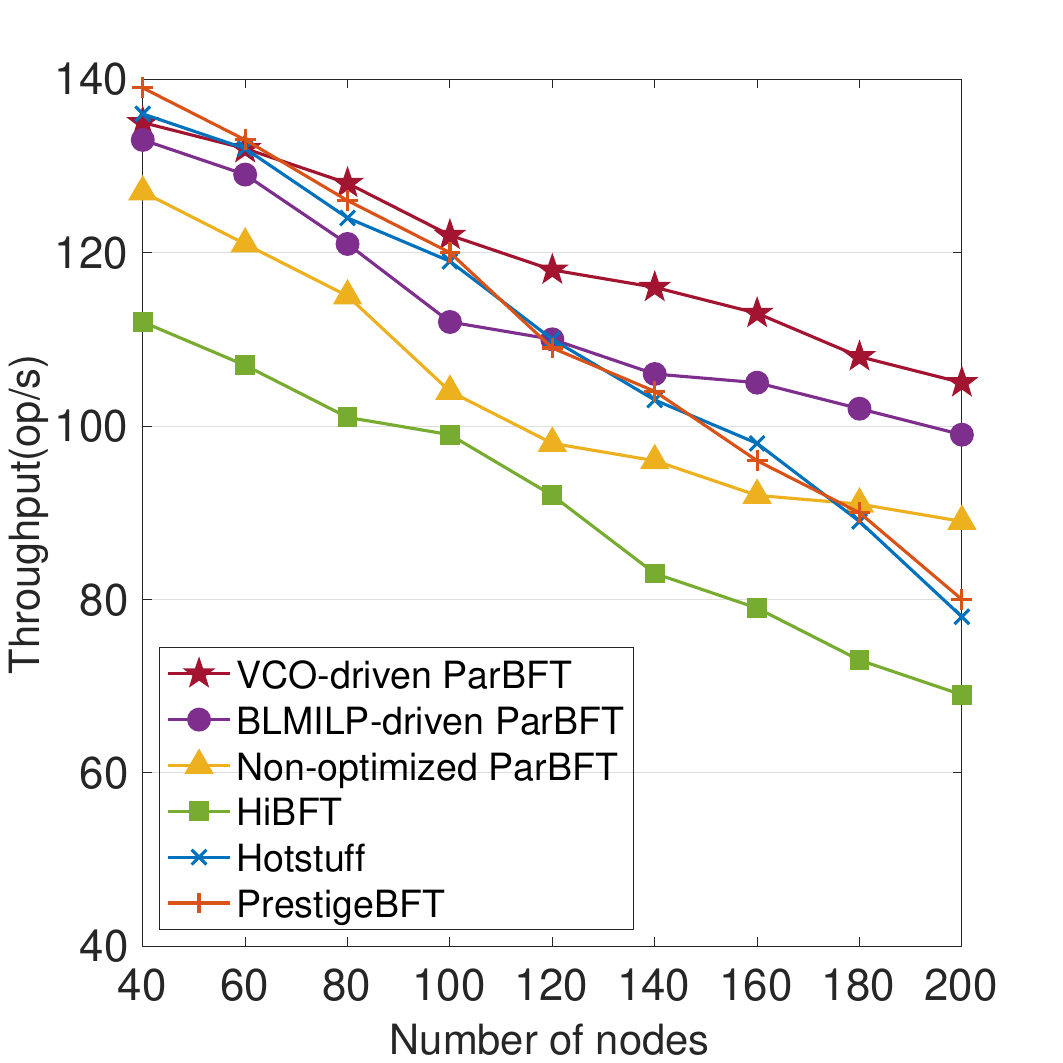}
        \caption{Throughput comparison}
        \label{subfig1}
    \end{subfigure}
    \begin{subfigure}[b]{0.48\textwidth}
        \centering
        \includegraphics[width=\textwidth]{ 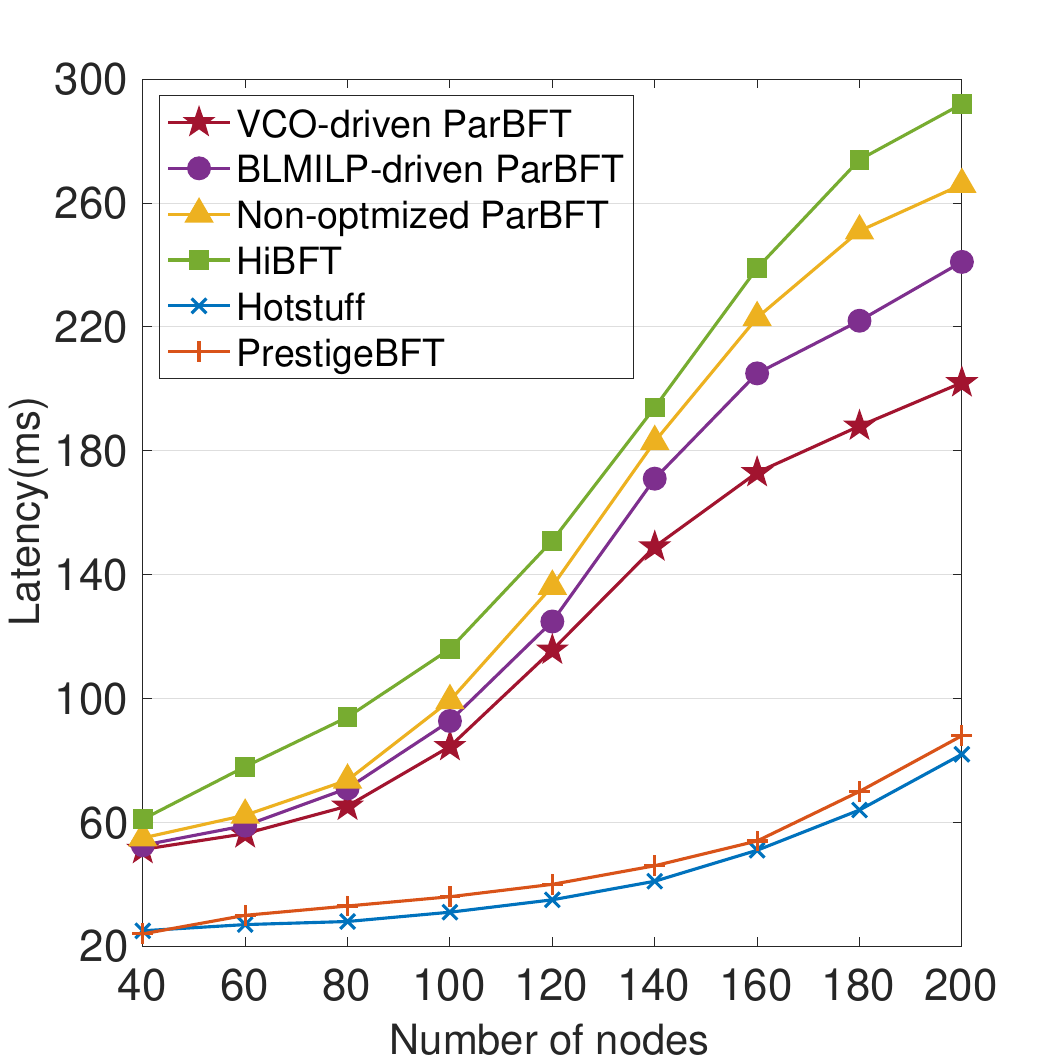}
        \caption{Latency comparison}
        \label{subfig2}
    \end{subfigure}
    \caption{Performance comparison of parallel BFT algorithms}
    \label{fig2}
\end{figure*}

Figure \ref{subfig1} presents the throughput comparison for ParBFT driven by our proposed VCO model with other parallel BFT algorithms, as the number of nodes increases from 40 to 200. Both the VCO-driven ParBFT and BL-MILP-driven ParBFT exhibit notable improvements in throughput compared to the non-optimized ParBFT configuration. Among the algorithms, VCO-driven ParBFT demonstrates the best performance, achieving an approximately 17.9\% improvement in throughput over the non-optimized ParBFT. As the number of nodes reaches 200, VCO-driven ParBFT further improves system throughput by about 6.1 \% compared to the BL-MILP-based ParBFT, highlighting the effectiveness of VCO in optimizing consensus performance under normal case operations.

Figure \ref{subfig2} illustrates the latency comparison for the same BFT algorithms, with node counts ranging from 40 to 200. VCO-driven ParBFT consistently achieves the lowest latency among the three algorithms. When the number of nodes reaches 200, ParBFT with VCO reduces latency by 24.1\% compared to the non-optimized ParBFT configuration. Overall, our proposed VCO model shows better performance in terms of both throughput and latency under normal case operation according to the experimental results.

\subsection{Performance under failure case}
We conducted experiments to evaluate the performance of VCO under failure conditions. To simulate a realistic failure environment, we assigned a subset of nodes a high probability of encountering extremely long message delays. These conditions were designed to represent crash failures and Byzantine failures, allowing us to test the performance of the ParBFT protocol under failure case and assess the effectiveness of VCO. The experiments compared the performance of VCO across four different configurations: VCO-driven ParBFT, SP-driven ParBFT, ParBFT with FD (Failure Detection), and non-optimized ParBFT.

The SP-driven ParBFT configuration utilizes a SP approach within the ParBFT framework \cite{xie2025stochastic}. 
We also examined ParBFT with FD, which integrates the standard ParBFT protocol with an enhanced Failure Detection mechanism \citep{er2022data}. This FD mechanism is crucial for  detecting crash failures by leveraging historical failure data maintained by the internal failure handler in ParBFT \citep{chandra1996unreliable}. The inclusion of FD is intended to reduce consensus delays caused by potential crashes, thereby improving overall system resilience.
As a baseline, we included a non-optimized ParBFT. This configuration serves as a performance benchmark, illustrating how the protocol behaves under failure conditions when no strategy are applied. 

\begin{figure*}[!htb]
    \centering
    \begin{subfigure}[b]{0.48\textwidth}
        \centering
        \includegraphics[width=\textwidth]{ 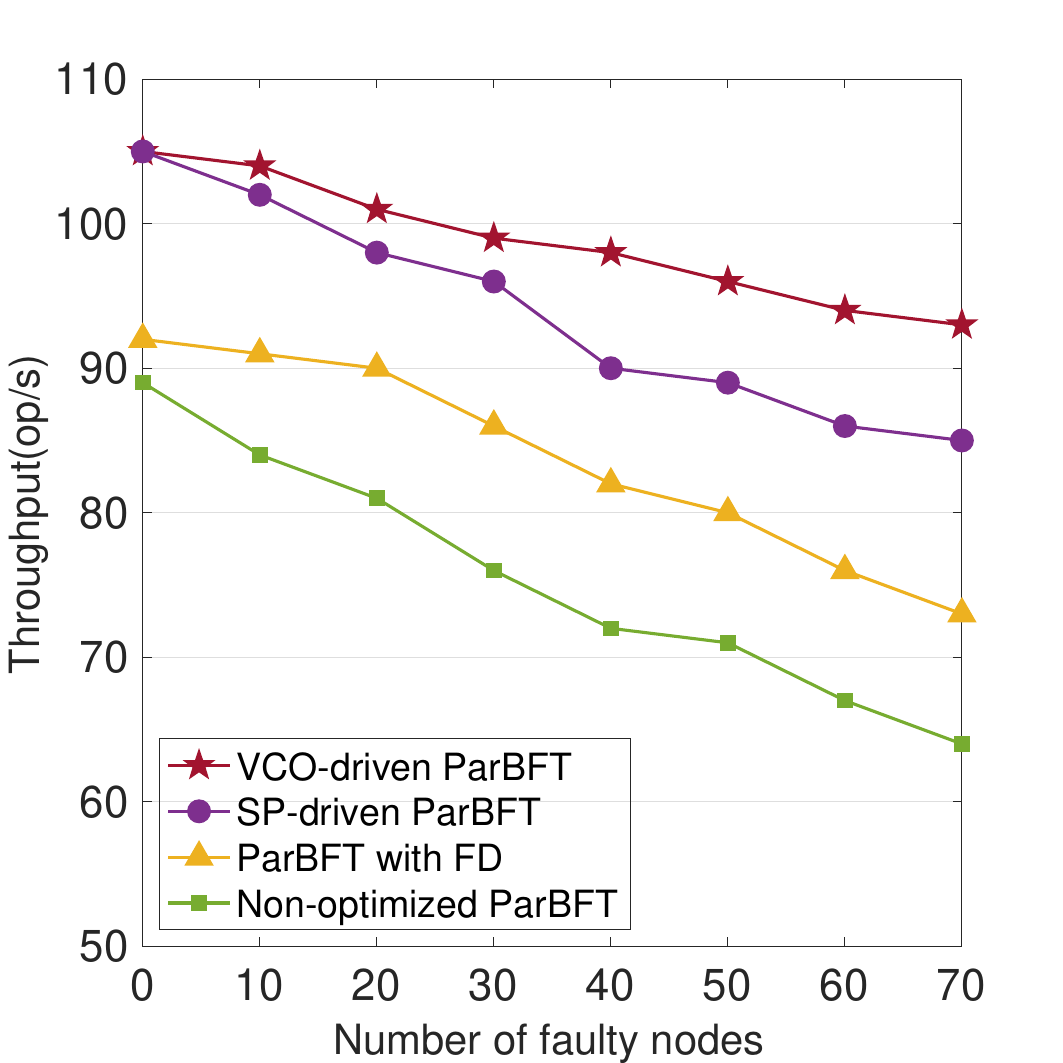}
        \caption{Throughput comparison}
        \label{subfig3}
    \end{subfigure}
    \begin{subfigure}[b]{0.48\textwidth}
        \centering
        \includegraphics[width=\textwidth]{ 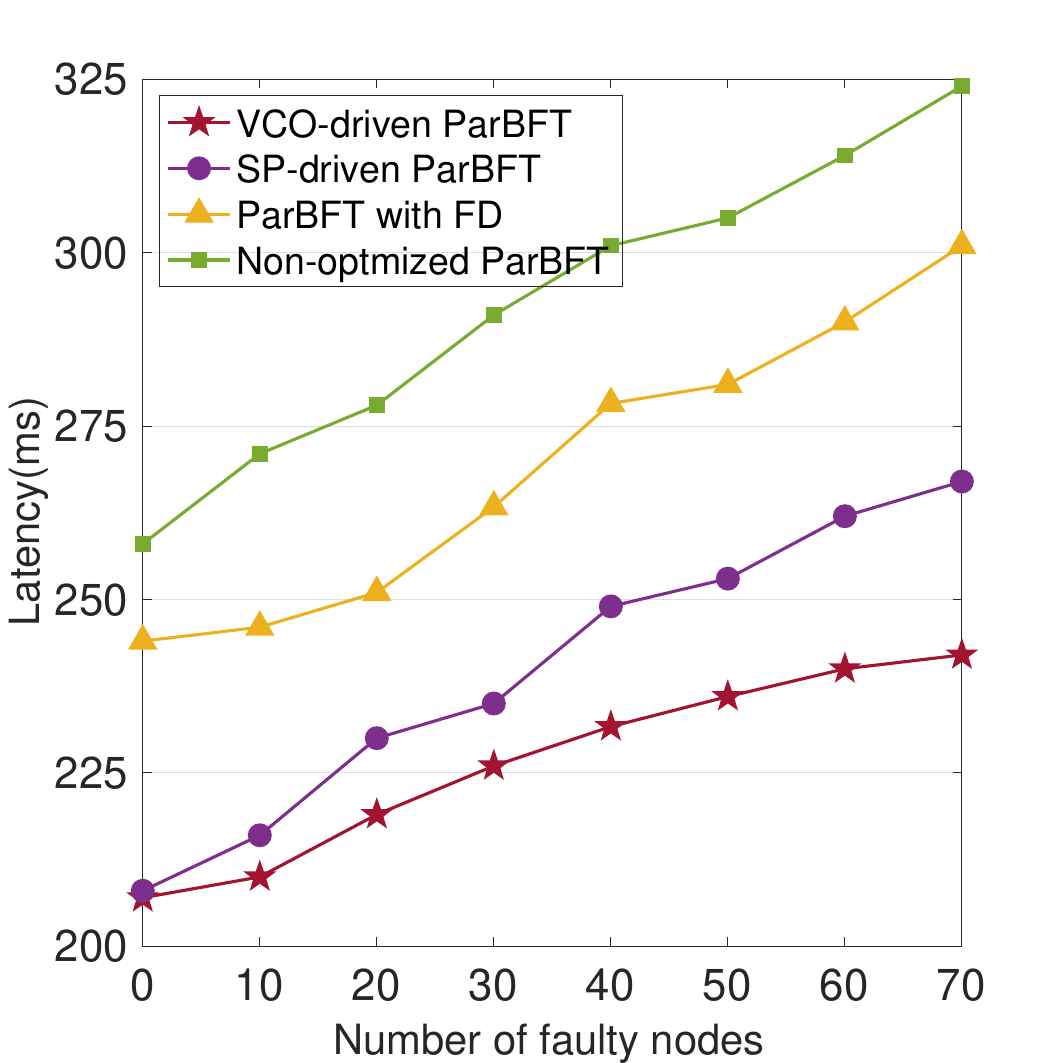}
        \caption{Latency comparison}
        \label{subfig4}
    \end{subfigure}
    \caption{Performance comparison of ParBFT algorithms under failure condition}
    \label{fig3}
\end{figure*}

Figure \ref{fig3} presents a comparison of throughput and latency between ParBFT with different configurations methods. As shown in Figure \ref{subfig3}, ParBFT with VCO achieves the highest throughput. When the number of faulty nodes increases from 0 to 70 out of a total of 200 nodes, ParBFT with VCO maintains both the highest throughput and the lowest latency. The improvements introduced by the VCO model compared with non-optimized ParBFT become increasingly significant as the number of faulty nodes grows. For instance, with 10 faulty nodes, throughput improves by 23.8\%, and with 70 faulty nodes, the improvement reaches 45.3\%. In terms of both throughput and latency, VCO-driven ParBFT outperforms configurations optimized by other models, such as SP and FD. Notably, as the proportion of faulty nodes increases, the performance of VCO-driven ParBFT degrades more slowly compared to the other configurations. This indicates that the VCO model provides superior resilience, particularly in scenarios with a higher number of faulty nodes, compared to models using FD or random view change reconfigurations.

\subsection{Performance comparison with various message sizes}

In this subsection, we compare the performance of VCO-driven ParBFT against BL-MILP-driven ParBFT and non-optimized ParBFT across different message sizes to assess how each model performs in varying operational environments. By varying the message size from 0 bytes to 1MB, we highlight the efficiency of VCO-driven ParBFT across a range of network operation conditions, showing that it achieves superior throughput and lower latency in comparison to both BL-MILP-driven and non-optimized ParBFT configurations.

\begin{figure*}[!htb]
    \centering
    \begin{subfigure}[b]{0.48\textwidth}
        \centering
        \includegraphics[width=\textwidth]{ 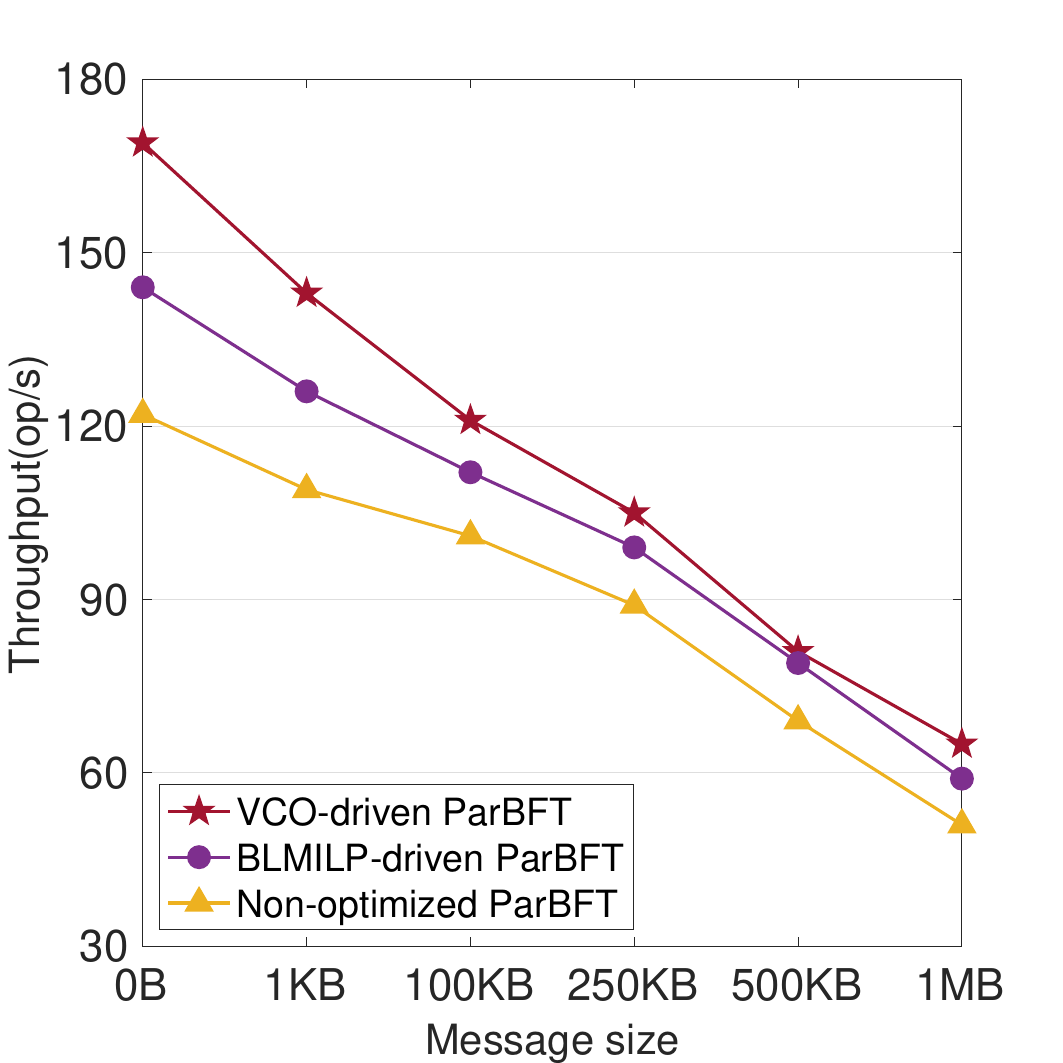}
        \caption{Throughput comparison}
        \label{subfig5}
    \end{subfigure}
    \begin{subfigure}[b]{0.48\textwidth}
        \centering
        \includegraphics[width=\textwidth]{ 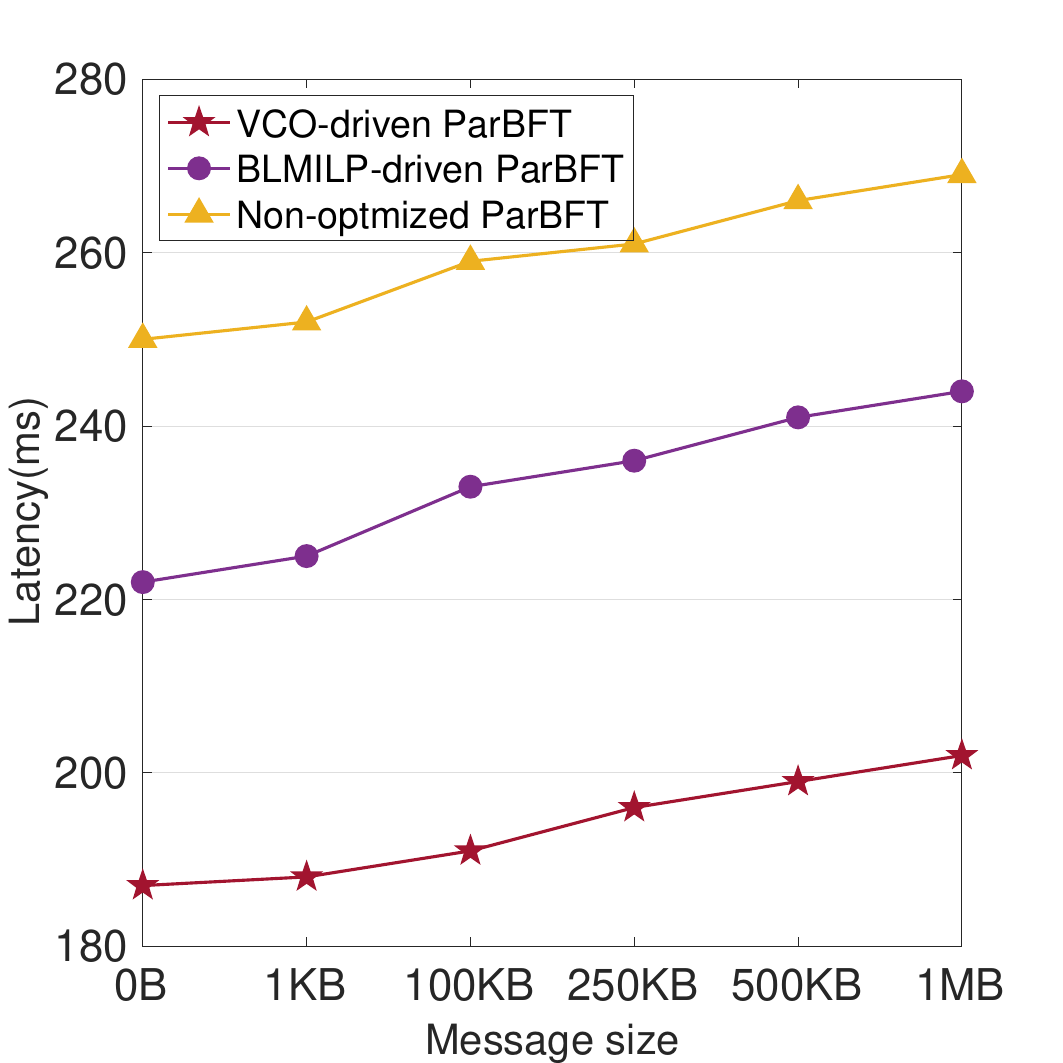}
        \caption{Latency comparison}
        \label{subfig6}
    \end{subfigure}
    \caption{Performance comparison of ParBFT algorithms with various message sizes}
    \label{fig4}
\end{figure*}

Figure \ref{subfig5} illustrates the throughput comparison for the three configurations. VCO-driven ParBFT outperforms both BL-MILP-driven and non-optimized ParBFT across all message sizes. At smaller message sizes (e.g., 1KB), VCO-driven ParBFT achieves the highest throughput. As message size increases, throughput decreases for all configurations due to the additional communication overhead introduced by larger messages. When the message sizes reach 1MB, VCO-driven ParBFT retains superior throughput compared to BL-MILP-driven and non-optimized ParBFT. 

Similarly, the latency comparison shown in Figure \ref{subfig6} shows that VCO-driven ParBFT exhibits the lowest latency across all message sizes. As message size increases, the latency for VCO-driven ParBFT rises more gradually compared to the other configurations. In contrast, BL-MILP-driven ParBFT shows a sharper increase in latency with larger messages, while the non-optimized ParBFT performs the worst, with latency increasing significantly as the message size exceeds 250KB.

\section{Conclusion}
In this paper, we address the inefficiency of the view change process in parallel BFT protocols. We introduce the VCO model to improve performance during both view changes and normal operation while preserving Byzantine safety guarantees. To the best of our knowledge, this is the first work that actively optimizes the view change process in parallel BFT systems, enabling more efficient leader selection and node reallocation under leader failures.
We further propose an iterative algorithm that specifies how backup leaders are selected as views proceed. To solve the VCO model efficiently, we adopt a decomposition-based approach with an efficient subproblem formulation and an improved Benders cut. By combining this efficient solution method with the VCO model, we design an efficient iterative backup leader selection algorithm that operates as views proceed.

Our experimental evaluation demonstrated that the VCO-driven ParBFT outperforms existing BL-MILP and SP optimization models under both normal operation and faulty condition. As network sizes increased and faulty nodes were introduced, the VCO model maintained superior performance, showing that it is well-suited for failure-resilient distributed systems. The results show that VCO for parallel BFT can significantly improve system efficiency and robustness.

In terms of future directions, the VCO model can be extended to explore its adaptability across a broader range of consensus algorithms beyond BFT. Additionally, integrating predictive techniques, such as machine learning-based failure prediction, could further enhance the model’s capacity to  manage leader changes and node reallocation.  By extending these aspects, the VCO model can serve as a foundation for the development of more resilient and scalable parallel consensus mechanisms, applicable to a wide variety of distributed systems, ranging from high-performance enterprise environments to more decentralized applications.

\bibliographystyle{unsrtnat}
\bibliography{mainref}

\end{document}